\documentclass[runningheads]{llncs}

\usepackage[T1]{fontenc}

\usepackage{graphicx}
\usepackage{subcaption}

\usepackage{amsmath,amssymb,mathtools,stmaryrd}

\usepackage{algorithm2e}
\RestyleAlgo{ruled}
\LinesNumbered
\SetKwComment{Comment}{/* }{ */}
\SetKw{Skip}{continue}
\SetKwFor{Match}{match}{with}{end}

\newcommand{\fail}{\texttt{\upshape Fail}}

\usepackage[hypertexnames=false,bookmarksdepth=2]{hyperref}
\usepackage{color}
\usepackage{tcolorbox}
\usepackage[mathlines]{lineno}

\renewcommand{\|}{\:|\:}
\newcommand{\Bool}{\texttt{Bool}}
\newcommand{\band}{\mathbin{\&}}
\newcommand{\True}{\texttt{True}}
\newcommand{\False}{\texttt{False}}
\newcommand{\ifexprtrail}[1]{\texttt{if}\,#1\,\texttt{then}\,\dots}
\newcommand{\ifexpr}[3]{\texttt{if}\,#1\,\texttt{then}\,#2\,\texttt{else}\,#3}

\newcommand{\myth}{\textsc{Myth}}
\newcommand{\nsync}{\textsc{NSynC}}

\newcommand{\unit}{\langle\rangle}

\newcommand{\fst}{\pi_1}
\newcommand{\snd}{\pi_2}
\newcommand{\pair}[2]{\langle #1,#2 \rangle}

\newcommand{\inl}[1]{\texttt{in}_L^{#1}}
\newcommand{\inr}[1]{\texttt{in}_R^{#1}}
\newcommand{\match}[3]{\delta(#1,\,#2,\,#3)}

\newcommand{\lam}[2]{\lambda (#1).#2}
\newcommand{\app}[2]{#1\,#2}

\newcommand{\UGamma}{\mathrm{\Gamma}}
\newcommand{\UXi}{\mathrm{\Xi}}
\newcommand{\exset}[3]{#1 : \mathrm{Ex}[#2;#3]}
\newcommand{\nref}[5]{\exset{#1}{#2}{#3},#4 \rightsquigarrow_\mathrm{N} #5}
\newcommand{\pref}[4]{\exset{#1}{#2}{#3} \rightsquigarrow_\mathrm{P} #4}
\newcommand{\mguess}[4]{[#1;#2],#3 \rightsquigarrow_\mathrm{M} #4}

\newtcbox{\mybox}{on line,
    sharp corners,
    colback=gray!15,
    colframe=black,
    boxrule=\fboxrule,
    boxsep=0pt,
    top=\fboxsep,
    bottom=\fboxsep,
    left=\fboxsep,
    right=\fboxsep
}
\renewcommand{\boxed}[1]{\mybox{\ensuremath{#1}}}

\spnewtheorem*{sketch}{Proof sketch}{\itshape}{\rmfamily}

\newcommand{\mysection}[1]{\par\textbf{#1.}\hspace{0.3em}}

\begin{document}

\title{NSynC: Normalised Synthesis of Computation}

\author{Zoey Shepherd
\and
Ohad Kammar
\and
Elizabeth Polgreen
}

\authorrunning{Z. Shepherd et al.}

\institute{University of Edinburgh, Edinburgh, UK}

\maketitle

\begin{abstract}

Inductive program synthesis algorithms search a space of programs to find one that meets some specification. Enumerating according to the syntax of a programming language leads to a large search space, and hence slow synthesis, due in large part to \textit{semantic duplication}. A synthesiser may have to evaluate---and reject---multiple semantically identical but syntactically different programs, wasting resources.

To avoid this duplication, we present \nsync{}, a synthesis-by-semantics approach. By enumerating the semantics of the target language directly, we guarantee that each candidate program is semantically unique and that each evaluation of a candidate is meaningful. Specifically, we search the space of normal forms for the simply-typed lambda calculus with sums using a top-down, type-directed synthesis algorithm. Our preliminary results show a geomean speedup of 8.93x on a synthetic benchmark suite over the unrestricted algorithm.

\end{abstract}

\noindent 
\mysection{Introduction} 
Program synthesis is the task of taking a program specification to a concrete program that meets it. Deductive approaches (e.g. \cite{manna1992fundamentals}) are limited in what programs they can produce, so search-based approaches like Counter-Example Guided Inductive Synthesis \cite{solar2006combinatorial} are more common for modern program synthesis \cite{alur2013syntax}. These search-based approaches search over \textit{syntax}. That is, if the syntax of the target language includes the production rule for Booleans $\Bool \Coloneqq \True\|\False\|a\|b\|\Bool\band\Bool\|\neg\Bool$, then a naïve synthesiser will consider the syntactically distinct terms $a\band b$, $b\band a$, $a\band(b\band \True)$, $\neg(\neg (a\band b))$, and so on to be distinct "guesses" for a target program of type \Bool. Of course, all of these programs are semantically identical, but a synthesiser might have to consider each of them individually to be able to discount them all.

Many synthesisers take steps to reduce the impact of semantic duplication. The standard formulation of bottom-up enumeration \cite{udupa2013transit} removes observationally equivalent programs from its pool at each step, preventing duplicates from surviving for more than one iteration and avoiding construction of terms like $a\band(b\band \True)$---but not, say, $a\band b$ and $b\band a$. In Semantics-Guided Synthesis \cite{kim2021semantics}, semantics are provided to the solver as a logical specification, making the process of ruling out duplicates easier but leaving them in the search space.

\mysection{Our Approach} To fully eliminate semantic duplication, an algorithm would search the space of \textit{semantics}, rather than syntax. We propose leveraging Normalisation by Evaluation, which constructs unique normal forms for semantic fragments of languages, to provide this search space. As a proof of concept, we present \nsync{} (Normalised Synthesis of Computation), which synthesises terms of the simply-typed lambda calculus with sums (STLC+) restricted to unique normal forms based on Balat et al. \cite{balat2004extensional}. \nsync{} uses the same top-down, type-directed framework as \myth{} \cite{osera2015type}.

STLC+ is a simple but powerful language fragment with a lot of semantic redundancy. It extends STLC, which has anonymous functions, with pairs, a unit type $1$ with sole value $\unit$, and sum types $\tau_1+\tau_2$ equipped with left and right injections $\inl{\tau_1+\tau_2}t$ and $\inr{\tau_1+\tau_2}t$ and match terms $\match{t}{x_1.t_1}{x_2.t_2}$ analogous to
\begin{align*}
    \texttt{match}\;t\;&\texttt{with} && &&\text{$t$ is the scrutinee} \\
    & \inl{}x_1 \Rightarrow t_1 && \text{where} &&\text{$t_1$ is the left branch, binding $x_1$} \\
    & \inr{}x_2 \Rightarrow t_2 && &&\text{$t_2$ is the right branch, binding $x_2$}.
\end{align*}
We can express Boolean logic using STLC+: type $\Bool \triangleq 1+1$, and terms $\True\triangleq\inl{\Bool}\unit$, $\False\triangleq\inr{\Bool}\unit$, $\ifexpr{t}{t_1}{t_2} \triangleq \match{t}{x_1.t_1}{x_2.t_2}$, and $a\band b \triangleq \ifexpr{a}{b}{\False}$.

\myth{} already eliminates some of the semantic duplication in STLC+ by synthesising only $\beta$-normal, $\eta$-long forms---where every computation is fully simplified and every term is expanded into its maximal structural form. This allows \myth{} to skip terms like $\ifexprtrail{\True}$, but it doesn't fully cover the semantics of match terms, which induce several semantic duplicates like $\ifexpr{a}{t}{t} = t$ or $a\band b=b\band a$.

Balat et al. \cite{balat2004extensional} account for both $\beta$-$\eta$ equivalences---as \myth{} does---and the strong sum extensionality axiom, which is responsible for equalities like above. Their normal forms eliminate redundant branches like $\ifexpr{a}{t}{t}$, and our extension with an ordering on scrutinees chooses between, e.g., $a\band b$ and $b \band a$, giving unique normal forms (proof in Appendix \ref{app:stlc_norm}). These normal forms define a fragment of STLC+ with full expressibility and no semantic duplication, giving an ideal search space for synthesis.

\nsync{}, like \myth{}, iteratively applies a set of synthesis rules. The derivation of these rules is the core of our approach: a direct application of \myth{} to STLC+ gets its rules directly from the language's definition, whereas \nsync{} draws them from the rules on normal forms. One such rule, for synthesising match terms, is
\begin{align*}
    \frac{
        \begin{array}{c}
            \mguess{\UGamma}{\tau_1+\tau_2}{\boxed{c}}{M} \\
            X_1 \triangleq \{ \sigma \cdot [v'/x_1] \mapsto v \| \sigma \mapsto v \in X, M[\sigma] \rightarrow^{*} \inl{\tau_1+\tau_2}(v') \} \\
            X_2 \triangleq \{ \sigma \cdot [v'/x_2] \mapsto v \| \sigma \mapsto v \in X, M[\sigma] \rightarrow^{*} \inr{\tau_1+\tau_2}(v') \} \\
            \forall i \in \{1,2\}.\; X_i \neq \emptyset \land \nref{X_i}{\UGamma,x_i:\tau_i}{\tau}{\boxed{(M \sqsubset)}}{N_i} \\
            \boxed{x_1 \notin FV(N_1) \land x_2 \notin FV(N_2) \implies N_1 \neq N_2}
        \end{array}
    }{
        \nref{X}{\UGamma}{\tau}{\boxed{c}}{\match{M}{x_1.N_1}{x_2.N_2}}
    }
\end{align*}
which begins with producing a scrutinee $M$ of sum type, then separates the example set $X$ accordingly to synthesise branches $N_1$ and $N_2$. The boxes highlight our divergence from \myth{}: the constraints $c$ and $M \sqsubset$ force scrutinees to appear in our chosen order, and the final precondition forces the two branches to be distinct. We show how our other rules differ from \myth{} in Appendix \ref{app:synth-rules}.

The \nsync{} algorithm exhaustively applies the synthesis rules until a solution emerges. To avoid non-termination in any subproblem, we impose size limits on neutral terms---that is, scrutinees and terms of base type---and a limit on the depth of match terms. The user inputs a search procedure over these limits, gradually increasing them until \nsync{} finds a satisfying term.

\mysection{Theoretical Results} We show \textit{bounded correctness} of \nsync{} in Appendix \ref{app:synth-alg}. Specifically, we show \textit{bounded completeness}: if the normal form of a term $t$ does not have any branching argument---a function that branches on its input and is given as an argument to a higher-order function---then \nsync{} can produce that normal form. Since \myth{} has the same restriction, we have that restricting to normal forms maintains solvability. We also show \textit{semantic optimality}: all of \nsync{}'s candidate terms are semantically distinct.

\mysection{Empirical Evaluation} We evaluate \nsync{} on a set of 166 randomly-generated synthetic benchmarks and see a geomean speedup of 8.93x. We compare \nsync{} with \myth{}*, our re-implementation of \myth{}. Figure \ref{fig:results-short} shows, for each of the 104 benchmarks that both algorithms solved, the time taken and the number of match terms enumerated by each algorithm. A further 37 benchmarks were solved by \myth{}* but not \nsync{}, since enforcing an order on scrutinees sometimes increases the number of branches in a program so that the example set is unable to cover them; in future work, we aim to address this problem using a parallel case split construction from Altenkirch et al. \cite{altenkirch2001normalization}.

\begin{figure}[t]
    \centering
    \begin{subfigure}[t]{0.49\textwidth}
        \centering
        \includegraphics[width=\linewidth]{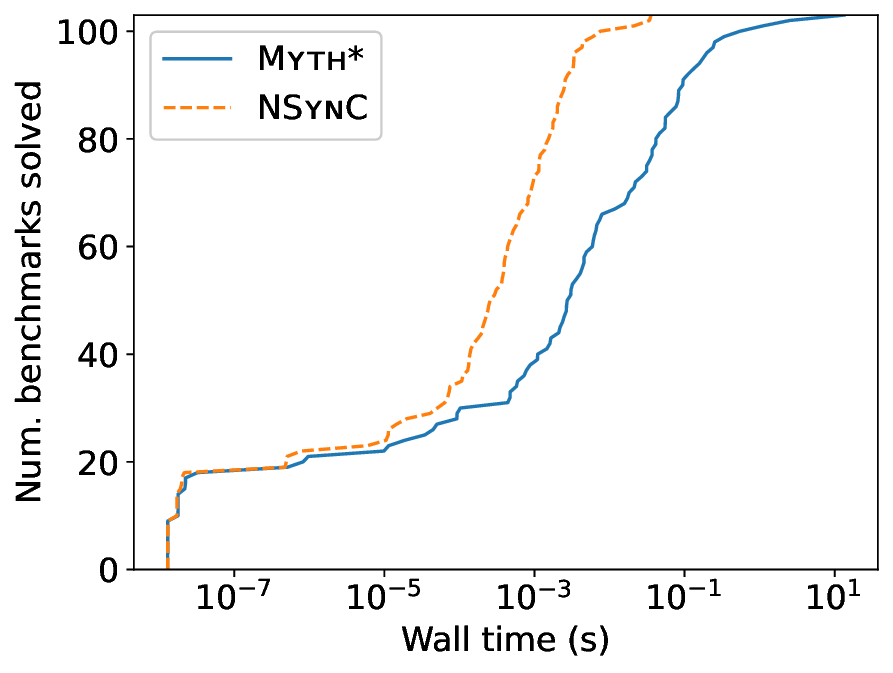}
    \end{subfigure}~\begin{subfigure}[t]{0.49\textwidth}
        \centering
        \includegraphics[width=\linewidth]{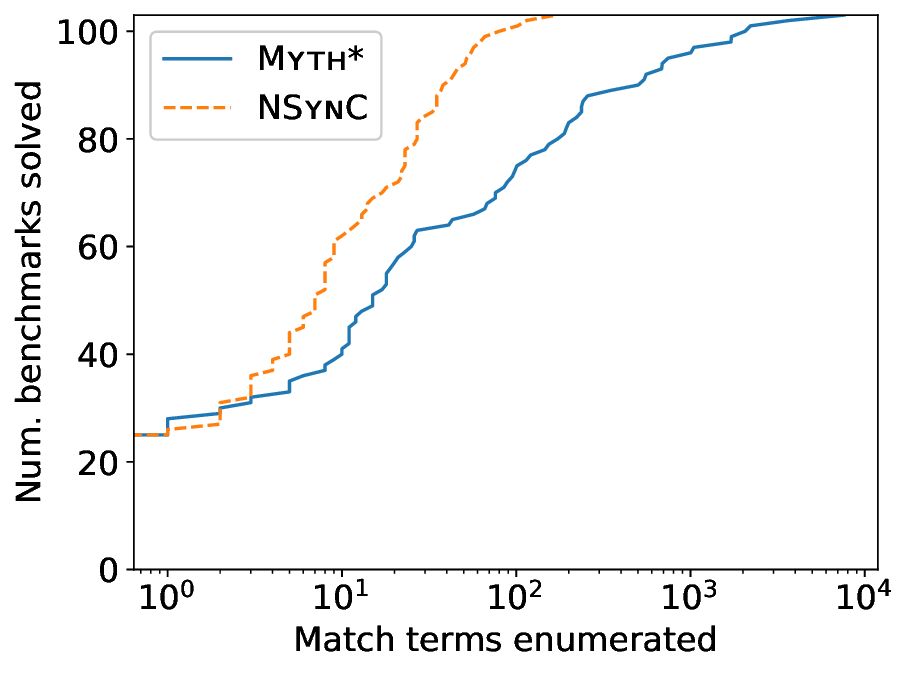}
    \end{subfigure}
    \caption{Log-scaled cactus plot showing the number of benchmarks that could be solved within time (left) and enumeration (right) constraints.}
    \label{fig:results-short}
\end{figure}

\mysection{Future Work} After addressing this limitation, the next step is to bring \nsync{} from a proof of concept to a practical synthesis method. In particular, we aim to extend our analysis to different language fragments that have normal forms, such as STLC with lists \cite{allais2013new}, or fragments corresponding to universal algebras \cite{yallop2018partially}. We then want to explore general-purpose synthesisers as combinations of these specialised, synthesis-by-semantics components.

\begin{credits}
\subsubsection{\ackname} Funding for this research was provided by EPSRC through a PhD studentship within the CDT in Machine Learning Systems hosted in the School of Informatics, University of Edinburgh (EP/Y03516X/1).

\subsubsection{\discintname} The authors have no competing interests to declare that are relevant to the content of this article.
\end{credits}

\bibliographystyle{splncs04}
\bibliography{mybibliography}

\clearpage\appendix

\section{STLC+}\label{app:stlc+}

The simply-typed lambda calculus with sums (STLC+) is the language over the context-free grammar shown in Figure \ref{fig:stlc+grammar}. $\theta_i$ denotes a base type, indexed by $i$. $1$ is the unit type, containing only the value $\unit$. $\tau_1\to\tau_2$ are function types, with $\lam{x:\tau_1}{t}$ as the constructor and $\app{t_1}{t_2}$ as the application of $t_1$ to $t_2$. $\tau_1 \times \tau_2$ are product types, constructed as $\pair{t_1}{t_2}$ and destructed with $\fst$ and $\snd$, which take the first and second element of the pair respectively. $\tau_1 + \tau_2$ are sum types, described in the main paper.

\begin{figure}[t]
    \centering
    \begin{subfigure}[t]{0.5\linewidth}
        \centering
        \begin{equation*}
            \begin{array}{rcll}
                \tau & {} \Coloneqq {} & \theta_i & \text{base types} \\
                &|& 1 & \text{unit type} \\
                &|& \tau\to\tau\quad & \text{function types} \\
                &|& \tau \times \tau & \text{product types} \\
                &|& \tau + \tau & \text{sum types}
            \end{array}
        \end{equation*}
        \begin{equation*}
            \begin{array}{rcl}
                \Bool &=& 1 + 1 \\
                \True &=& \inl{\Bool}\unit \\
                \False &=& \inr{\Bool}\unit \\
                \ifexpr{t}{t_1}{t_2} &=& \match{t}{x_1.t_1}{x_2.t_2}
            \end{array}
        \end{equation*}
    \end{subfigure}\begin{subfigure}[t]{0.5\linewidth}
        \centering
        \begin{equation*}
            \begin{array}{rcll}
                t & {}\Coloneqq{} & x & \text{variables} \\
                &|& \unit & \text{unit} \\
                &|& \lam{x:\tau}{t} & \text{abstraction} \\
                &|& \app{t}{t} & \text{application} \\
                &|& \pair{t}{t} & \text{pairing} \\
                &|& \fst t & \text{first projection} \\
                &|& \snd t & \text{second projection} \\
                &|& \inl{\tau+\tau} t & \text{left injection} \\
                &|& \inr{\tau+\tau} t & \text{right injection} \\
                &|& \match{t}{x.t}{x.t}\quad & \text{match term}
            \end{array}
        \end{equation*}
    \end{subfigure}
    \caption{Grammar for types (upper left) and terms (right) of STLC+, where $i$ ranges over some index set $I$ of base types and $x$ ranges over variables. We also give syntactic sugar (lower left) for Boolean values.}
    \label{fig:stlc+grammar}
\end{figure}

Most formulations of STLC+ (including Balat et al.'s \cite{balat2004extensional}) also have an empty type $0$ which contains no values, with the semantics that finding a term of the empty type is a contradiction from which any other type can be derived. By building our types from the grammar in Figure \ref{fig:stlc+grammar}, we can avoid any risk of this contradiction arising.

This appendix gives the typing rules (Figure \ref{fig:type-rules}), operational semantics (Figure \ref{fig:eval-rules}), and equational theory (Subappendix \ref{app:stlc_equ}) for STLC+.

We write $\UGamma \vdash t:\tau$ for the typing judgement, saying that in typing context $\UGamma$ the term $t$ is well-typed at $\tau$. A typing context $\UGamma \vdash \cdot \| \UGamma,x:\tau$ is simply a list assigning types to free variables, and we assume that the variables in a given context are unique.

We give operational semantics in terms of a relation $t \to^* v$, meaning a term $t$ evaluates to value $v$. Values are simply fully reduced closed terms, given by the grammar $v \Coloneqq a \| \unit \| \lam{x:\tau}{t} \| \pair{v}{v} \| \inl{\tau+\tau}v \| \inr{\tau+\tau}v$, where $a$ stands for any constant and the remaining symbols are the constructors in the term language.

Note that there is no rule for evaluation under a $\lambda$-abstraction. This is because we define equality between function-typed values using extensional equality: when $\cdot \vdash v_1,v_2:\tau'\to\tau$, we have $v_1 = v_2$ if and only if, for all values $\cdot \vdash v:\tau'$, we have $\app{v_1}{v} \to^* w_1$ and $\app{v_2}{v} \to^* w_2$ s.t. $w_1=w_2$.

\begin{figure}[p]
    \centering
    \begin{align*}
        \frac{x:\tau \in \UGamma}{\UGamma \vdash x:\tau}\:\textsc{TypeVar} &&
        \frac{}{\UGamma \vdash \unit:1}\:\textsc{TypeUnit}
    \end{align*}
    \begin{align*}
        \frac{\UGamma,x:\tau_1 \vdash t : \tau}{\UGamma \vdash \lam{x:\tau_1}{t} : \tau_1 \to \tau}\:\textsc{TypeAbs} &&
        \frac{
            \begin{array}{ccc}
                \UGamma \vdash t_1 : \tau_1 \to \tau &\quad & \UGamma \vdash t_2 : \tau_1
            \end{array}
        }{\UGamma \vdash \app{t_1}{t_2}:\tau}\:\textsc{TypeApp}
    \end{align*}
    \begin{align*}
        \frac{
            \begin{array}{ccc}
                \UGamma \vdash t_1 : \tau_1 &\quad& \UGamma \vdash t_2 : \tau_2
            \end{array}
        }{\UGamma \vdash \pair{t_1}{t_2}:\tau_1 \times \tau_2}\:\textsc{TypePair}
    \end{align*}
    \begin{align*}
        \frac{\UGamma \vdash t : \tau_1 \times \tau_2}{\UGamma \vdash \fst t : \tau_1}\:\textsc{TypeFst} &&
        \frac{\UGamma \vdash t : \tau_1 \times \tau_2}{\UGamma \vdash \snd t : \tau_2}\:\textsc{TypeSnd}
    \end{align*}
    \begin{align*}
        \frac{\UGamma \vdash t : \tau_1}{\UGamma \vdash \inl{\tau_1+\tau_2}(t) : \tau_1 + \tau_2}\:\textsc{TypeInL} &&
        \frac{\UGamma \vdash t : \tau_2}{\UGamma \vdash \inr{\tau_1+\tau_2}(t) : \tau_1 + \tau_2}\:\textsc{TypeInR}
    \end{align*}
    \begin{align*}
        \frac{
        \begin{array}{ccccc}
            \UGamma \vdash t : \tau_1 + \tau_2 &\quad& \UGamma,x_1:\tau_1 \vdash t_1 : \tau &\quad& \UGamma,x_2:\tau_2 \vdash t_2 : \tau
        \end{array}
        }{\UGamma \vdash \match{t}{x_1.t_1}{x_2.t_2} : \tau}\:\textsc{TypeMatch}
    \end{align*}
    \caption{Typing rules for STLC+.}
    \label{fig:type-rules}
    
    \centering
    \begin{align*}
        \frac{\cdot \vdash v:\tau}{v \to^* v}\:\textsc{EvalValue} &&
        \frac{
            \begin{array}{ccc}
                t_1 \to^* \lam{x:\tau}{t'} & \quad & t'[t_2/x] \to^* v
            \end{array}
        }{\app{t_1}{t_2} \to^* v}\:\textsc{EvalApp}
    \end{align*}
    \begin{align*}
        \frac{
            \begin{array}{cc}
                t_1 \to^* v_1 & t_1 \to^* v_2
            \end{array}
        }{\pair{t_1}{t_2} \to^* \pair{v_1}{v_2}}\:\textsc{EvalPair}
    \end{align*}
    \begin{align*}
        \frac{t \to^* \pair{v_1}{v_2} }{\fst t \to^* v_1}\:\textsc{EvalFst} &&
        \frac{t \to^* \pair{v_1}{v_2} }{\snd t \to^* v_2}\:\textsc{EvalSnd}
    \end{align*}
    \begin{align*}
        \frac{t \to^* v}{\inl{\tau_1+\tau_2}t \to^* \inl{\tau_1+\tau_2}v}\:\textsc{EvalInL} &&
        \frac{t \to^* v}{\inr{\tau_1+\tau_2}t \to^* \inr{\tau_1+\tau_2}v}\:\textsc{EvalInR}
    \end{align*}
    \begin{align*}
        \frac{
            \begin{array}{ccc}
                t \to^* \inl{\tau_1+\tau_2}(v') &\quad& t_1[v'/x_1] \to^* v
            \end{array}
        }{\match{t}{x_1.t_1}{x_2.t_2} \to^* v}\:\textsc{EvalMatchL}
    \end{align*}
    \begin{align*}
        \frac{
            \begin{array}{ccc}
                t \to^* \inr{\tau_1+\tau_2}(v') &\quad& t_2[v'/x_2] \to^* v
            \end{array}
        }{\match{t}{x_1.t_1}{x_2.t_2} \to^* v}\:\textsc{EvalMatchR}
    \end{align*}
    \caption{Operational semantics of STLC+.}
    \label{fig:eval-rules}
\end{figure}

\subsection{Semantic Equality}\label{app:stlc_equ}

We write $\UGamma \vdash t_1 \simeq_\tau t_2$ to mean that $t_1$ and $t_2$ are semantically equal at type $\tau$ in context $\UGamma$; we may choose to omit the type for brevity. The operational semantics above satisfies this equational theory: if $\UGamma \vdash t_1 \simeq_\tau t_1$, then for all environments $\UGamma \vdash \sigma$ assigning to each variable in $\UGamma$ a value of the appropriate type, $t_1[\sigma] \to^* v_1$ and $t_2[\sigma] \to^* v_2$ s.t. $v_1 = v_2$.

Semantic equality is inductively defined as the reflexive, symmetric, transitive, and congruence closure of the rules shown in Figure \ref{fig:eq-rules}.

\begin{figure}[t]
    \centering
    \begin{align*}
        \frac{\UGamma \vdash t:1}{\UGamma \vdash t \simeq_1 \unit}\:\textsc{EqUnit}
    \end{align*}
    \begin{align*}
        \frac{
            \begin{array}{ccc}
                \UGamma,x:\tau_1 \vdash t:\tau &\quad&
                \UGamma \vdash t_1:\tau_1
            \end{array}
        }{
            \UGamma \vdash \app{(\lam{x:\tau_1}{t})}{t_1} \simeq_\tau t[t_1/x]
        }\:\textsc{EqFunc$\beta$}
    \end{align*}
    \begin{align*}
        \frac{\UGamma \vdash t:\tau_1 \to \tau}{\UGamma \vdash \lam{x:\tau_1}{(\app{t}{x})} \simeq_{\tau_1\to\tau} t}\:\textsc{EqFunc$\eta$}
    \end{align*}
    \begin{align*}
        \frac{
            \begin{array}{ccc}
                \UGamma \vdash t_1:\tau_1 &\quad& \UGamma \vdash t_2:\tau_2
            \end{array}
        }{\UGamma \vdash \fst \pair{t_1}{t_2} \simeq_{\tau_1} t_1}\:\textsc{EqProd$\beta_1$} &&
        \frac{
            \begin{array}{ccc}
                \UGamma \vdash t_1:\tau_1 &\quad& \UGamma \vdash t_2:\tau_2
            \end{array}
        }{\UGamma \vdash \snd \pair{t_1}{t_2} \simeq_{\tau_2} t_2}\:\textsc{EqProd$\beta_2$}
    \end{align*}
    \begin{align*}
        \frac{\UGamma \vdash t:\tau_1 \times \tau_2}{\UGamma \vdash \pair{\fst t}{\snd t} \simeq_{\tau_1\times\tau_2} t}\:\textsc{EqProd$\eta$}
    \end{align*}
    \begin{align*}
        \frac{
            \begin{array}{ccccc}
                \UGamma \vdash t:\tau_1 &\quad& \UGamma,x_1:\tau_1 \vdash t_1:\tau &\quad& \UGamma,x_2:\tau_2 \vdash t_2:\tau
            \end{array}
        }{\UGamma \vdash \match{\inl{\tau_1+\tau_2}t}{x_1.t_1}{x_2.t_2} \simeq_\tau t_1[t/x_1]}\:\textsc{EqSum$\beta_L$}
    \end{align*}
    \begin{align*}
        \frac{
            \begin{array}{ccccc}
                \UGamma \vdash t:\tau_2 &\quad& \UGamma,x_1:\tau_1 \vdash t_1:\tau &\quad& \UGamma,x_2:\tau_2 \vdash t_2:\tau
            \end{array}
        }{\UGamma \vdash \match{\inr{\tau_1+\tau_2}t}{x_1.t_1}{x_2.t_2} \simeq_\tau t_2[t/x_1]}\:\textsc{EqSum$\beta_R$}
    \end{align*}
    \begin{align*}
        \frac{
            \begin{array}{ccc}
                \UGamma \vdash t:\tau_1 + \tau_2 &\quad& \UGamma,x:\tau_1+\tau_2 \vdash t':\tau
            \end{array}
        }{\UGamma \vdash \match{t}{x_1.t'[\inl{\tau_1+\tau_2}x_1/x]}{x_2.t'[\inr{\tau_1+\tau_2}x_2/x]} \simeq_\tau t'[t/x] }\:\textsc{EqSSEA}
    \end{align*}
    \caption{An equational theory for STLC+.}
    \label{fig:eq-rules}
\end{figure}

The first rule \textsc{EqUnit} is simply that the unit type only contains one value. The rest are familiar $\beta$-$\eta$ laws for each type, except for \textsc{EqSSEA}, the Strong Sum Extensionality Axiom \cite{balat2004extensional}. \textsc{EqSSEA} generalises the $\eta$ law for sum types (recovered by setting $t' = x$) with the added property that execution contexts distribute over match terms.

\begin{example}
We define the following syntactic sugar:
\begin{equation*}
    \begin{array}{rcl}
        \Bool &=& 1 + 1 \\
        \True &=& \inl{\Bool}\unit \\
        \False &=& \inr{\Bool}\unit \\
        \ifexpr{t}{t_1}{t_2} &=& \match{t}{x_1.t_1}{x_2.t_2} \\
        t_1\&t_2 &=& \ifexpr{t_1}{t_2}{\False}
    \end{array}
\end{equation*}
and show that the axioms above derive $a:\Bool,b:\Bool \vdash  a \& b \simeq_\Bool b\& a$:
\begin{align*}
    a\&b &= \match{a}{x_1.b}{x_2.\inr{\Bool}\unit} \\
    &= \match{a}{x_1.\match{b}{y_1.\inl{\Bool}y_1}{y_2.\inr{\Bool}y_2}}{x_2.\match{b}{y_1.\inr{\Bool}\unit}{y_2.\inr{\Bool}\unit}} \\
    &= \match{a}{x_1.\match{b}{y_1.\inl{\Bool}x_1}{y_2.\inr{\Bool}\unit}}{x_2.\match{b}{y_1.\inr{\Bool}x_2}{y_2.\inr{\Bool}\unit}}
\end{align*}
The first equality is just desugaring, and the second is by two applications of \textsc{EqSSEA}; one to $\eta$-expand the left branch, the other to introduce on the right branch a match with identical branches. We then apply \textsc{EqUnit} to re-label the unit-valued terms, and that allows us to set $t' = \match{b}{y_1.x}{y_2.\inr{\Bool}\unit}$ and apply \textsc{EqSSEA} once more to arrive at our result:
\begin{align*}
    \match{a}{x_1.t'[\inl{\Bool}x_1/x]}{x_2.t'[\inr{\Bool}x_2/x]} &= t'[a/x] \\
    &= \match{b}{y_1.a}{y_2.\inr{\Bool}\unit} \\
    &= b \& a
\end{align*}
\end{example}

\section{Normal Forms of STLC+}\label{app:stlc_norm}

We begin by presenting the normal forms of STLC+ given by Balat et al. \cite{balat2004extensional}. As noted in their paper, these normal forms are not unique. For example,
\begin{align*}
    \ifexpr{a}{(\ifexpr{b}{\True}{\False})}{\False}, &&\text{and} \\
    \ifexpr{b}{(\ifexpr{a}{\True}{\False})}{\False}.
\end{align*}
are both normal forms of $a \band b$. They suggest that fixing an ordering on scrutinees would give unique normal forms; we present our ordering in Subappendix \ref{app:ordering} and show in Subappendix \ref{app:uniqueness} that it does, in fact, give uniqueness.

The grammar of normal forms is given in Figure \ref{fig:norm-grammar}, where the semantics of the terminal symbols are unchanged; note that this, and that $P \Coloneqq M$ only at base type, already gives $\beta$-long, $\eta$-short forms. $M$ are neutral terms, those where computation is stuck at a free variable, $P$ are pure normal terms, and $N$ are normal terms. What we call neutral terms are referred to by Balat et al. as "pure neutral", distinguishing from "impure" neutral terms which can include match statements. However, with $N$ as the start symbol, there is no way to reach an "impure" neutral in the grammar, so we do not need to consider them.

We write $\UGamma \vdash_\mathrm{M} t:\tau$ to denote that $t$ is a neutral term at type $\tau$ in context $\UGamma$, and similarly for $\UGamma \vdash_\mathrm{P} \dots$ and $\UGamma \vdash_\mathrm{N} \dots$. The complete rules for normal forms \cite[Figure~3]{balat2004extensional} are given in Figure \ref{fig:norm-rules}; most are identical to the corresponding typing rules, with the exception of \textsc{NFNeutral}---which is restricted to base type---and \textsc{NFAbs*} and \textsc{NFMatch*}.

\begin{figure}[p]
    \centering
    \begin{align*}
        \tau &\Coloneqq \theta_i\|1\|\tau\to\tau\|\tau\times\tau\|\tau+\tau \\
        M &\Coloneqq x\|\app{M}{P}\|\fst M\|\snd M \\
        P &\Coloneqq M\|\unit\|\lam{x:\tau}{N}\|\pair{P}{P}\|\inl{\tau+\tau}P\|\inr{\tau+\tau}P \\
        N &\Coloneqq P\|\match{M}{x.N}{x.N}
    \end{align*}
    \caption{Grammar for normal forms of STLC+ \cite{balat2004extensional}.}
    \label{fig:norm-grammar}
    
    \centering
    \begin{align*}
        \frac{x:\tau \in \UGamma}{\UGamma \vdash_\mathrm{M} x:\tau}\:\textsc{NFVar} &&
        \frac{
            \begin{array}{ccc}
                \UGamma \vdash_\mathrm{M} M : \tau_1 \to \tau &\quad & \UGamma \vdash_\mathrm{P} P : \tau_1
            \end{array}
        }{\UGamma \vdash_\mathrm{M} \app{M}{P}:\tau}\:\textsc{NFApp}
    \end{align*}
    \begin{align*}
        \frac{\UGamma \vdash_\mathrm{M} M : \tau_1 \times \tau_2}{\UGamma \vdash_\mathrm{M} \fst M : \tau_1}\:\textsc{NFFst} &&
        \frac{\UGamma \vdash_\mathrm{M} M : \tau_1 \times \tau_2}{\UGamma \vdash_\mathrm{M} \snd M : \tau_2}\:\textsc{NFSnd}
    \end{align*}
    \begin{align*}
        \frac{\exists i.\UGamma \vdash_\mathrm{M} M:\theta_i}{\UGamma \vdash_\mathrm{P} M:\theta_i}\:\textsc{NFNeutral} &&
        \frac{}{\UGamma \vdash_\mathrm{P} \unit:1}\:\textsc{NFUnit}
    \end{align*}
    \begin{equation*}
        \frac{
            \begin{array}{ccc}
                \UGamma,x:\tau_1 \vdash_\mathrm{N} N : \tau &\quad& \forall C \in \text{Guards}(N).x \in FV(C)
            \end{array}
        }{\UGamma \vdash_\mathrm{P} \lam{x:\tau_1}{N}:\tau_1\to\tau}\:\textsc{NFAbs*}
    \end{equation*}
    \begin{align*}
        \frac{
            \begin{array}{ccc}
                \UGamma \vdash_\mathrm{P} P_1 : \tau_1 &\quad& \UGamma \vdash_\mathrm{P} P_2 : \tau_2
            \end{array}
        }{\UGamma \vdash_\mathrm{P} \pair{P_1}{P_2}:\tau_1 \times \tau_2}\:\textsc{NFPair}
    \end{align*}
    \begin{align*}
        \frac{\UGamma \vdash_\mathrm{P} P : \tau_1}{\UGamma \vdash_\mathrm{P} \inl{\tau_1+\tau_2}(P) : \tau_1 + \tau_2}\:\textsc{NFInL} &&
        \frac{\UGamma \vdash_\mathrm{P} P : \tau_2}{\UGamma \vdash_\mathrm{P} \inr{\tau_1+\tau_2}(P) : \tau_1 + \tau_2}\:\textsc{NFInR}
    \end{align*}
    \begin{align*}
        \frac{\UGamma \vdash_\mathrm{P} P : \tau}{\UGamma \vdash_\mathrm{N} P:\tau}\:\textsc{NFPure}
    \end{align*}
    \begin{align*}
        \frac{
            \begin{array}{c}
                \UGamma \vdash_\mathrm{M} M:\tau_1+\tau_2 \\
                \forall i \in \{1,2\}.\; \UGamma,x_i:\tau_i \vdash_\mathrm{N} N_i : \tau \land{} \\
                \forall C \in \text{Guards}(x_i.N_i).M \neq C\\
                x_1 \notin FV(N_1) \land x_2 \notin FV(N_2) \implies N_1 \neq N_2
            \end{array}
        }{\UGamma \vdash_\mathrm{N} \match{M}{x_1.N_1}{x_2.N_2} : \tau}\:\textsc{NFMatch*}
    \end{align*}
    \begin{align*}
        \text{Guards}(N) &\triangleq \begin{cases}
            \{M\} \cup \bigcup_{i \in \{1,2\}} \text{Guards}(x_i.N_i) & \text{if}\;N=\match{M}{x_1.N_1}{x_2.N_2} \\
            \emptyset & \text{otherwise}
        \end{cases} \\
        \text{Guards}(x_i.N_i) &\triangleq \{ C \in \text{Guards}(N) \| x_i \notin FV(C) \}.
    \end{align*}
    \caption{Normal form rules for STLC+ \cite{balat2004extensional}. $FV(t)$ denotes the free variables that appear in a term $t$.}
    \label{fig:norm-rules}
\end{figure}

\subsection{Ordering of Neutrals}\label{app:ordering}

To get unique normal forms, we replace the guard restriction in \textsc{NFMatch*} with a stronger restriction in terms of an ordering on neutral terms, which we define in this appendix. This also allows us to rephrase the guard restriction in \textsc{NFAbs*} in a similar way.

We write $\UGamma \vdash M_1 \sqsubset M_2$ to mean that $M_1$ is strictly before $M_2$ in context $\UGamma$. Fixing $\UGamma$ gives us a binary relation $\UGamma \vdash - \sqsubset -$ on neutral terms, which is a total ordering. We define an auxiliary total ordering $\UGamma \vdash - \sqsubset_\tau -$ on pure normal terms that share the type $\tau$, and give mutually inductive definitions of the two (Figure \ref{fig:ord-rules}). Since we cannot synthesise neutral terms that contain match terms---a limitation we inherit from \myth{}, see the main paper---we omit the ordering on those here.

We first define some additional notation. $\UGamma_1 \leq \UGamma$ denotes that the typing context $\UGamma_1$ is a prefix of $\UGamma$, that is $\UGamma = \UGamma_1,\UGamma_2$ for some $\UGamma_2$. $\UGamma \vdash t$ denotes that $t$ is well-defined under $\UGamma$, similar to $\UGamma \vdash t:\tau$ but where $\tau$ is irrelevant. $\UGamma \vDash t_1,t_2$ denotes that $t_1$ and $t_2$ are equally defined under $\UGamma$, meaning that for all $\UGamma_1 \leq \UGamma$ we have $\UGamma_1 \vdash t_1 \iff \UGamma_1 \vdash t_2$. $M_1 \sqsubset_d M_2$ is a transitive relation between neutrals based only on the top-level destructor, that is
\begin{equation*}
    x \sqsubset_d \fst t \sqsubset_d \snd t \sqsubset_d \app{t}{t}
\end{equation*}
where $x$ is any variable and each $t$ represents any term.

Both of these relations, with fixed $\UGamma$ and, if applicable, $\tau$, are irreflexive, antisymmetric, transitive, and total. Additionally, extending a context $\UGamma$ with a new variable $x$ preserves the ordering on already-definable terms, places all newly-definable terms above them, and gives $x$ as the least of the newly-definable terms. Using these properties, we can replace the rules \textsc{NFAbs*} and \textsc{NFMatch*} with those shown in Figure \ref{fig:norm-rules-ordered} (changes highlighted), and see that this does not weaken the preconditions.

\begin{figure}
    \centering
    \begin{align*}
        \frac{
            \begin{array}{ccccc}
                \UGamma_1 \leq \UGamma &\quad& \UGamma_1 \vdash M_1 &\quad& \UGamma_1 \not\vdash M_2
            \end{array}
        }{\UGamma \vdash M_1 \sqsubset M_2}\:\textsc{OrdCtx}
    \end{align*}
    \begin{align*}
        \frac{
            \begin{array}{ccc}
                \UGamma \vDash M_1,M_2 &\quad& M_1 \sqsubset_d M_2
            \end{array}
        }{\UGamma \vdash M_1 \sqsubset M_2}\:\textsc{OrdOp}
    \end{align*}
    \begin{align*}
        \frac{\UGamma \vdash M_1 \sqsubset M_2}{\UGamma \vdash \fst M_1 \sqsubset \fst M_2}\:\textsc{OrdFst} &&
        \frac{\UGamma \vdash M_1 \sqsubset M_2}{\UGamma \vdash \snd M_1 \sqsubset \snd M_2}\:\textsc{OrdSnd}
    \end{align*}
    \begin{align*}
        \frac{
            \begin{array}{ccc}
                \UGamma \vDash \app{M_1}{P_1},\app{M_2}{P_2} &\quad& \UGamma \vdash M_1 \sqsubset M_2
            \end{array}
        }{\UGamma \vdash \app{M_1}{P_1}\sqsubset\app{M_2}{P_2}}\:\textsc{OrdAppFunc}
    \end{align*}
    \begin{align*}
        \frac{
            \begin{array}{ccccc}
                \UGamma \vDash \app{M}{P_1},\app{M}{P_2} &\quad& \UGamma \vdash M:\tau_1\to\tau &\quad& \UGamma \vdash P_1 \sqsubset_{\tau_1} P_2
            \end{array}
        }{\UGamma \vdash \app{M}{P_1}\sqsubset\app{M}{P_2}}\:\textsc{OrdAppArg}
    \end{align*}
    \begin{align*}
        \frac{
            \begin{array}{ccccc}
                \UGamma_1 \leq \UGamma &\quad& \UGamma_1 \vdash P_1:\tau &\quad& \UGamma_1 \not\vdash P_2
            \end{array}
        }{\UGamma \vdash P_1 \sqsubset_\tau P_2}\:\textsc{OrdCtxPure}
    \end{align*}
    \begin{align*}
        \frac{
            \begin{array}{ccc}
                \UGamma \vdash M_1,M_2:\theta_i &\quad& \UGamma \vdash M_1 \sqsubset M_2
            \end{array}
        }{\UGamma \vdash M_1 \sqsubset_{\theta_i} M_2}\:\textsc{OrdNeutral}
    \end{align*}
    \begin{align*}
        \frac{x:\tau_1,\UGamma \vdash P_1 \sqsubset_{\tau} P_2}{\UGamma \vdash \lam{x:\tau_1}{P_1} \sqsubset_{\tau_1\to\tau} \lam{x:\tau_1}{P_2}}\:\textsc{OrdAbs}
    \end{align*}
    \begin{align*}
        \frac{
            \begin{array}{ccc}
                \UGamma\vDash\pair{P_{11}}{P_{12}},\pair{P_{21}}{P_{22}} &\quad& \UGamma \vdash P_{11} \sqsubset_{\tau_1} P_{21}
            \end{array}
        }{\UGamma \vdash \pair{P_{11}}{P_{12}}\sqsubset_{\tau_1\times\tau_2}\pair{P_{21}}{P_{22}}}\:\textsc{OrdPair1}
    \end{align*}
    \begin{align*}
        \frac{
            \begin{array}{ccc}
                \UGamma\vDash\pair{P}{P_{12}},\pair{P}{P_{22}} &\quad& \UGamma \vdash P_{12} \sqsubset_{\tau_2} P_{22}
            \end{array}
        }{\UGamma \vdash \pair{P}{P_{12}}\sqsubset_{\tau_1\times\tau_2}\pair{P}{P_{22}}}\:\textsc{OrdPair2}
    \end{align*}
    \begin{align*}
        \frac{
            \UGamma \vDash P_1,P_2
        }{\UGamma \vdash \inl{\tau_1+\tau_2} P_1 \sqsubset_{\tau_1+\tau_2} \inr{\tau_1+\tau_2} P_2}\:\textsc{OrdInjections}
    \end{align*}
    \begin{align*}
        \frac{\UGamma \vdash P_1 \sqsubset_{\tau_1} P_2}{\UGamma \vdash \inl{\tau_1+\tau_2} P_1 \sqsubset_{\tau_1+\tau_2} \inl{\tau_1+\tau_2} P_2}\:\textsc{OrdInL}
    \end{align*}
    \begin{align*}
        \frac{\UGamma \vdash P_1 \sqsubset_{\tau_2} P_2}{\UGamma \vdash \inr{\tau_1+\tau_2} P_1 \sqsubset_{\tau_1+\tau_2} \inr{\tau_1+\tau_2} P_2}\:\textsc{OrdInR}
    \end{align*}
    \caption{Rules for our ordering relations.}
    \label{fig:ord-rules}
\end{figure}

\begin{figure}
    \centering
    \begin{equation*}
        \frac{
            \begin{array}{c}
                \UGamma,x:\tau_1 \vdash_\mathrm{N} N : \tau \\
                \boxed{N = \delta(M,\dots) \implies \UGamma,x:\tau_1 \vdash x \sqsubseteq M}
            \end{array}
        }{\UGamma \vdash_\mathrm{P} \lam{x:\tau_1}{N}:\tau_1\to\tau}\:\textsc{NFAbs}
    \end{equation*}
    \begin{align*}
        \frac{
            \begin{array}{c}
                \UGamma \vdash_\mathrm{M} M:\tau_1+\tau_2 \\
                \forall i \in \{1,2\}.\; \UGamma,x_i:\tau_i \vdash_\mathrm{N} N_i : \tau \land{} \\
                \boxed{N_i = \delta(M_i, \dots) \implies \UGamma,x:\tau_i \vdash M \sqsubset M_i} \\
                x_1 \notin FV(N_1) \land x_2 \notin FV(N_2) \implies N_1 \neq N_2
            \end{array}
        }{\UGamma \vdash_\mathrm{N} \match{M}{x_1.N_1}{x_2.N_2} : \tau}\:\textsc{NFMatch}
    \end{align*}
    \caption{Our modified normal form rules, with our ordering on scrutinees.}
    \label{fig:norm-rules-ordered}
\end{figure}

\subsection{Unique Normal Forms}\label{app:uniqueness}

Balat et al.'s construction gives every term a normal form \cite[Theorem~5.1]{balat2004extensional}. With our described changes, we can show the stronger condition that every term has a \textit{unique} normal form; rather than building directly on Balat et al.'s theorem, we construct our proof by way of Altenkirch et al.'s normal forms \cite{altenkirch2001normalization}.

\begin{theorem}[Uniqueness of Normal Forms]\label{thm:normalisation}
    For every STLC+ term $\UGamma \vdash t:\tau$, there exists a unique normal term $\UGamma \vdash_\mathrm{N} N:\tau$ such that $\UGamma \vdash t \simeq_\tau N$. Therefore, if $\UGamma \vdash_\mathrm{N} N_1,N_2 : \tau$, then $N_1 = N_2 \iff \UGamma \vdash N_1 \simeq_\tau N_2$.
\end{theorem}
\begin{proof}
    Follows from Altenkirch et al.'s results \cite{altenkirch2001normalization} and that their reification function \cite[Definition~3.2]{altenkirch2001normalization}, when scrutinees are chosen by our ordering, is an injective mapping from their normal forms to ours. \qed
\end{proof}

\begin{corollary}
    If $\UGamma \vdash_\mathrm{P} P_1,P_2 : \tau$, then $P_1=P_2$ if and only if $\UGamma \vdash P_1 \simeq_\tau P_2$. If $\UGamma \vdash_\mathrm{M} M_1,M_2 : \tau$, then $M_1=M_2$ if and only if $\UGamma \vdash M_1 \simeq_\tau M_2$.
\end{corollary}

\section{NSynC}\label{app:nsync}

Since we do not wish to provide concrete functions to our synthesiser as example outputs---this would amount to giving the synthesiser the answer---we extend the definition of values to include \textit{partial functions}. A partial function $\cdot \vdash \tilde{f}:\tau_1\to\tau$ relates distinct values of type $\tau_1$ (the set of which is written $\text{dom}(\tilde{f})$) to values of type $\tau$. Extensional equality relates partial functions to each other and to concrete functions when they agree on all inputs for which both are defined.

The input specification for our synthesis algorithm is a \textit{typed example set} $\exset{X}{\UGamma}{\tau}$, consisting of input-output examples $\sigma \mapsto v$. The input $\sigma$ is an assignment of values of appropriate type to the variables in $\UGamma$. The output is a value $\cdot \vdash v:\tau$ which may only contain partial, not concrete, functions.

We also define \textit{neutral constraints} $c \Coloneqq M \sqsubset \| M \sqsubseteq \| \top$, and write $\UGamma \vdash c(M')$ to mean that $M'$ satisfies $c$ in context $\UGamma$: $\UGamma \vdash (M \sqsubset)(M')$ and $\UGamma \vdash (M \sqsubseteq)(M')$ expand in the obvious way, and $\UGamma \vdash \top(M')$ is always true. Constraints let us place lower bounds on term enumeration, and enforce the ordering on neutrals.

\subsection{Overview}\label{app:synth-overview}

Normal forms give us a duplication-free search space for synthesis. No two normal forms have the same semantics, and yet every possible semantics is represented in some normal form, so any term we could have synthesised without normalisation has an equivalent that we can synthesise now.

Since our target language is typed and purely functional, the type-directed synthesis algorithm \myth{} \cite{osera2015type} is a natural choice for synthesising it. Our synthesis begins with a top-down, recursive process of generating constructors based on the goal type and example sets. When we are unable to generate a constructor---either at base type, or when examples include both left and right injections---we use an enumerative search to find a neutral term that solves the current subproblem or a match term that partitions the example set in such a way that synthesis can progress.

The algorithm halts when it finds a satisfying term. To prevent it from enumerating indefinitely, we also impose a limit on the sizes of neutral terms the algorithm can produce. An outer loop then iteratively raises that limit until the top-down search can find a satisfying term, or a time limit is reached.

Our algorithm, with \myth{}, is defined in terms of synthesis relations which govern how the algorithm's inputs relate to its required output. We present these relations in Subappendix \ref{app:synth-rules} and the algorithm that arises from them in Subappendix \ref{app:synth-alg}; we also present correctness theorems alongside each formulation.

\subsection{Synthesis Relations}\label{app:synth-rules}

The three synthesis relations, corresponding to neutral, pure normal, and normal terms respectively, are written
\begin{align*}
    [\UGamma ; \tau],c &\rightsquigarrow_\mathrm{M} M && \text{($M$-guess, Figure \ref{fig:m-guess-rules})} \\
    \exset{X}{\UGamma}{\tau} &\rightsquigarrow_\mathrm{P} P && \text{($P$-refine, Figure \ref{fig:p-refine-rules}), and} \\
    \exset{X}{\UGamma}{\tau},c &\rightsquigarrow_\mathrm{N} N && \text{($N$-refine, Figure \ref{fig:n-refine-rules}),}
\end{align*}
and defined by mutual induction. Each relates a specification to the terms of the appropriate sort that satisfy it, and therefore that we wish the synthesiser to produce.

The $M$-guess relation enumerates all neutral terms $\UGamma \vdash_\mathrm{M} M:\tau$ satisfying $c$; since neutrals are stuck at free variables, there is no way to use an example set to narrow down the search, so it is omitted from the relation. $P$-refinement deductively synthesises constructors based on the examples and the goal type, e.g. a goal of product type induces a pair constructor. $N$-refinement synthesises match terms, with the constraint applied to the scrutinee, and also allows producing a pure neutral if no more matches are needed.

The distinction between $P$- and $N$-refinement and the boxed sections of the rules do not appear in the corresponding relations of \myth{}.

\begin{figure}[t]
    \centering
    \begin{align*}
        \frac{
            \begin{array}{ccc}
                x:\tau \in \UGamma &\quad& \boxed{\UGamma \vdash c(x)}
            \end{array}
        }{\mguess{\UGamma}{\tau}{\boxed{c}}{x}}\:\textsc{MGuessVar}
    \end{align*}
    \begin{align*}
        \frac{
            \begin{array}{ccccc}
                \mguess{\UGamma}{\tau'\to\tau}{\boxed{\top}}{M} &\quad& \pref{\emptyset}{\UGamma}{\tau'}{P} &\quad& \boxed{\UGamma \vdash c(\app{M}{P})}
            \end{array}
        }{\mguess{\UGamma}{\tau}{\boxed{c}}{\app{M}{P}}}\:\textsc{MGuessApp}
    \end{align*}
    \begin{align*}
        \frac{
            \begin{array}{ccc}
                \mguess{\UGamma}{\tau\times \tau'}{\boxed{\top}}{M} &\quad& \boxed{\UGamma \vdash c(\fst M)}
            \end{array}
        }{\mguess{\UGamma}{\tau}{\boxed{c}}{\snd M}}\:\textsc{MGuessFst}
    \end{align*}
    \begin{align*}
        \frac{
            \begin{array}{ccc}
                \mguess{\UGamma}{\tau'\times \tau}{\boxed{\top}}{M} &\quad& \boxed{\UGamma \vdash c(\snd M)}
            \end{array}
        }{\mguess{\UGamma}{\tau}{\boxed{c}}{\snd M}}\:\textsc{MGuessSnd}
    \end{align*}
    \caption{Enumeration rules for $M$-guess.}
    \label{fig:m-guess-rules}
\end{figure}
\begin{figure}
    \centering
    \begin{align*}
        \frac{
            \begin{array}{ccc}
                \mguess{\UGamma}{\theta_i}{\boxed{\top}}{M} &\quad& \forall \sigma \mapsto v \in X.\: M[\sigma] \to^* v
            \end{array}
        }{\pref{X}{\UGamma}{\theta_i}{M}}\:\textsc{PRefineNeutral}
    \end{align*}
    \begin{align*}
        \frac{}{\pref{X}{\UGamma}{1}{\unit}}\:\textsc{PRefineUnit}
    \end{align*}
    \begin{align*}
        \frac{
            \begin{array}{c}
                X_1 \triangleq\{(\sigma \cdot [v/x]) \mapsto \tilde{f}(v) \| \sigma \mapsto \tilde{f} \in X,\,v \in \text{dom}(\tilde{f})\} \\
                \nref{X_1}{\UGamma,x:\tau_1}{\tau}{\boxed{(x \sqsubseteq)}}{N}
            \end{array}
        }{\pref{X}{\UGamma}{\tau_1\to\tau}{\lam{x:\tau_1}{N}}}\:\textsc{PRefineAbs}
    \end{align*}
    \begin{align*}
        \frac{
            \begin{array}{c}
                \pref{\{ \sigma \mapsto v_1 \| \sigma \mapsto\pair{v_1}{v_2} \in X \}}{\UGamma}{\tau_1}{P_1} \\
                \pref{\{ \sigma \mapsto v_2 \| \sigma \mapsto\pair{v_1}{v_2} \in X \}}{\UGamma}{\tau_2}{P_2}
            \end{array}
        }{\pref{X}{\UGamma}{\tau_1\times\tau_2}{\pair{P_1}{P_2}}}\:\textsc{PRefinePair}
    \end{align*}
    \begin{align*}
        \frac{
            \begin{array}{c}
                \nexists \sigma \mapsto \inr{\tau_1+\tau_2}v \in X \\
                \pref{\{ \sigma \mapsto v \| \sigma \mapsto \inl{\tau_1+\tau_2}v \in X \}}{\UGamma}{\tau_1}{P}
            \end{array}
        }{\pref{X}{\UGamma}{\tau_1+\tau_2}{\inl{\tau_1+\tau_2}P}}\:\textsc{PRefineInL}
    \end{align*}
    \begin{align*}
        \frac{
            \begin{array}{c}
                \nexists \sigma \mapsto \inl{\tau_1+\tau_2}v \in X \\
                \pref{\{ \sigma \mapsto v \| \sigma \mapsto \inr{\tau_1+\tau_2}v \in X \}}{\UGamma}{\tau_2}{P}
            \end{array}
        }{\pref{X}{\UGamma}{\tau_1+\tau_2}{\inr{\tau_1+\tau_2}P}}\:\textsc{PRefineInR}
    \end{align*}
    \caption{Refinement rules for $P$-refine.}
    \label{fig:p-refine-rules}
\end{figure}
\begin{figure}
    \centering
    \begin{align*}
        \frac{\pref{X}{\UGamma}{\tau}{P}}{\nref{X}{\UGamma}{\tau}{c}{P}}\:\textsc{NRefinePure}
    \end{align*}
    \begin{align*}
        \frac{
            \begin{array}{c}
                \mguess{\UGamma}{\tau_1+\tau_2}{\boxed{c}}{M} \\
                X_1 \triangleq \{ \sigma \cdot [v'/x_1] \mapsto v \| \sigma \mapsto v \in X, M[\sigma] \rightarrow^{*} \inl{\tau_1+\tau_2}(v') \} \\
                X_2 \triangleq \{ \sigma \cdot [v'/x_2] \mapsto v \| \sigma \mapsto v \in X, M[\sigma] \rightarrow^{*} \inr{\tau_1+\tau_2}(v') \} \\
                \forall i \in \{1,2\}.\; X_i \neq \emptyset \land \nref{X_i}{\UGamma,x_i:\tau_i}{\tau}{\boxed{(M \sqsubset)}}{N_i} \\
                \boxed{x_1 \notin FV(N_1) \land x_2 \notin FV(N_2) \implies N_1 \neq N_2}
            \end{array}
        }{
            \nref{X}{\UGamma}{\tau}{\boxed{c}}{\match{M}{x_1.N_1}{x_2.N_2}}
        }\:\textsc{NRefineMatch}
    \end{align*}
    \caption{Refinement rules for $N$-refine.}
    \label{fig:n-refine-rules}
\end{figure}

\begin{example}
To illustrate these rules, consider synthesising a term with the semantics $a \land b$, where $a$ and $b$ are Boolean values. We begin with examples
\begin{equation*}
    X = \left\{\begin{array}{cllccl}
        [ & \True/a, & \True/b &] & \mapsto & \True, \\{}
        [ & \True/a, & \False/b &] & \mapsto & \False, \\{}
        [ & \False/a, & \True/b &] & \mapsto & \False, \\{}
        [ & \False/a, & \False/b &] & \mapsto & \False
    \end{array}\right\}
\end{equation*}
with type $\exset{X}{\UGamma}{\tau}$, where $\UGamma = a:\Bool,b:\Bool$ and $\tau = \Bool$. To synthesise a solution, we look for some $N$ such that $\nref{X}{\UGamma}{\tau}{\top}{N}$.

$N$-refine has two rules: \textsc{NRefinePure} and \textsc{NRefineMatch}. We cannot apply \textsc{NRefinePure}, since there are no rules for $P$-refine that apply to our example set; the goal type only fits \textsc{PRefineInL} and \textsc{PRefineInR}, and our example outputs contain a mixture of left and right injections. To apply \textsc{NRefineMatch}, we search for neutral terms of sum type satisfying our constraint (which is currently $\top$); there are two such terms, $a$ and $b$.

Say we choose $a$. We then define the two new example sets, $X_1$ and $X_2$:
\begin{align*}
    X_1 &= \left\{\begin{array}{clllccl}
        [ & \True/a, & \True/b & \unit/x_1 &] & \mapsto & \True, \\{}
        [ & \True/a, & \False/b & \unit/x_1 &] & \mapsto & \False
    \end{array}\right\} \\
    X_2 &= \left\{\begin{array}{clllccl}
        [ & \False/a, & \True/b & \unit/x_2 &] & \mapsto & \False, \\{}
        [ & \False/a, & \False/b & \unit/x_2 &] & \mapsto & \False
    \end{array}\right\}.
\end{align*}
Neither is empty, so we try to synthesise branches. Starting with the left branch, we want $N_1$ s.t. $\nref{X_1}{\UGamma,x_1:1}{\Bool}{(a \sqsubset)}{N_1}$. Our example outputs contain both left and right injections, so all we have is \textsc{NRefineMatch}; the only scrutinee we can guess is $b$, since $\UGamma \not\vdash (a \sqsubset)(a)$. We split the examples:
\begin{align*}
    X_{11} &= \{ [\True/a, \True/b,\unit/x_1,\unit/x_{11}] \mapsto \True \} \\
    X_{12} &= \{ [\True/a, \False/b,\unit/x_1,\unit/x_{12}] \mapsto \False \}.
\end{align*}
Again, neither is empty: we recurse to the left branch. The only example has a left injection as output, so we apply \textsc{NRefinePure} using \textsc{PRefineInL}. That gives us a new example set $\{[\True/a, \True/b,\unit/x_1,\unit/x_{11}] \mapsto \unit\}$ with a goal type of $1$, so we use \textsc{PRefineUnit} to synthesise $\unit$. Backtracking and doing the same for the right branch with $X_{12}$, we get $N_1 = \match{b}{x_{11}.\inl{1+1}\unit}{x_{12}.\inr{1+1}\unit}$ or, sugared, $N_1 = \ifexpr{b}{\True}{\False}$.

On the other branch, we see that we can immediately apply \textsc{NRefinePure} through \textsc{PRefineInR}, and get $N_2 = \inr{1+1}\unit$. We could not have applied \textsc{NRefineMatch} here, since we can only enumerate $M=b$ and find $N_{21}=N_{22}=\inr{1+1}\unit$, breaking the last precondition of \textsc{NRefineMatch}.

Backtracking again gives us the full solution,
\begin{equation*}
    N = \ifexpr{a}{(\ifexpr{b}{\True}{\False})}{\False}.
\end{equation*}
\end{example}

\subsubsection{Correctness.} We show that the synthesis relations are sound (Theorem \ref{thm:relations-soundness}), in the sense that they only synthesise terms satisfying the example set, and semantically optimal (Theorem \ref{thm:relations-optimality}), meaning only semantically distinct terms can be synthesised.

Unfortunately, our algorithm is not complete. We inherit from \myth{} the requirement when synthesising match terms that the subset of examples going to each branch be non-empty. This requirement prevents synthesising arbitrary match terms that do not further the synthesis goal and ensures termination when the goal can be satisfied, but comes at the cost of blocking certain terms: since $M$-guessing doesn't propagate examples, \nsync{} can't enumerate neutrals that contain matches as subterms. Since by the grammar such a neutral must be or contain an application term, we call them terms with \textit{branching arguments}. Theorem \ref{thm:relations-completeness} shows that our synthesis relations are complete for terms without branching arguments.

\myth{} can also synthesise all of the normal forms that \nsync{} can, with the same restriction on branching arguments. Therefore, as long as normalisation does not introduce branching arguments---we believe that this holds, though the proof requires reasoning about the normalisation procedure itself---restricting our synthesis to normal terms does not reduce completeness.

\begin{theorem}[Soundness]\label{thm:relations-soundness}
    If $\nref{X}{\UGamma}{\tau}{\top}{N}$ then $N$ satisfies $X$, that is for all $\sigma \mapsto v \in X$, $N[\sigma] \to^* v$.
\end{theorem}
\begin{proof}
    In \textsc{PRefineNeutral}, satisfaction of the example set is checked explicitly, and for \textsc{PRefineUnit} it is trivial; that the other refinement rules are sound follows by induction. \qed
\end{proof}

\begin{theorem}[Semantic Optimality]\label{thm:relations-optimality}
    If $\nref{X}{\UGamma}{\tau}{c}{N_1}$ and $\nref{X}{\UGamma}{\tau}{c}{N_2}$, then $\UGamma \vdash N_1 \simeq_\tau N_2$ if and only if $N_1 = N_2$, and similarly for $\pref{X}{\UGamma}{\tau}{P_1,P_2}$ and $\mguess{\UGamma}{\tau}{c}{M_1,M_2}$.
\end{theorem}
\begin{proof}
    The synthesis relations erase to the normal form rules (Subappendix \ref{app:stlc_norm}), so the terms produced by $N$-refine are normal (by $P$-refine are pure normal, and by $M$-guess are neutral), then by Theorem \ref{thm:normalisation} (and its corollary). \qed
\end{proof}

To prove bounded completeness for $N$-refine, we first need to show bounded completeness for $M$-guessing via the following lemma.

\begin{lemma}
    For any neutral term $\UGamma \vdash_\mathrm{M} M:\tau$ with no branching arguments, and any constraint $c$ s.t. $\UGamma \vdash c(M)$, $\mguess{\UGamma}{\tau}{c}{M}$.
\end{lemma}
\begin{proof}
    All four $M$-guess rules check the constraint explicitly rather than propagating it, so we can ignore the constraint in the rest of the proof and check that $\UGamma \vdash_\mathrm{M} M:\tau$ implies $\mguess{\UGamma}{\tau}{\top}{M}$ where $M$ has no branching arguments. This is just induction on the structure of $M$, with the observation that having an empty example set and no match terms renders the synthesis and normalisation rules almost identical. \qed
\end{proof}

\begin{theorem}[Bounded Completeness]\label{thm:relations-completeness}
    For any term $\UGamma \vdash t:\tau$ whose normal form $N$ has no branching arguments, there is at least one example set $X$ s.t. $\nref{X}{\UGamma}{\tau}{\top}{N}$.
\end{theorem}
\begin{sketch}
    We borrow from Altenkirch et al. \cite{altenkirch2001normalization} the notion of a neutral constrained environment $\UGamma|\UXi$, which gives a context $\UGamma$ a set $\Xi$ of constraints of the form $M = \inl{\tau_1+\tau_2}x_1$ or $M = \inr{\tau_1+\tau_2}x_2$, where $\UGamma \vdash_\mathrm{M} M:\tau_1+\tau_2$ and $x_i:\tau_i \in \UGamma$.

    We can recurse on the structure of $N$ to find a neutral constrained environment for each branch. If $\UGamma|\UXi \vdash \match{M}{x_1.N_1}{x_2.N_2}$, then we get $\UGamma,x_1:\tau_1|\UXi,M=\inl{\tau_1+\tau_2}x_1 \vdash N_1$ and $\UGamma,x_2:\tau_2|\UXi,M=\inr{\tau_1+\tau_2}x_2 \vdash N_2$. Other term-level symbols may duplicate the environment into branches (pairs) or add to the typing context (abstractions) as usual, but do not modify $\UXi$.

    Once each branch's neutral constrained environment has been found, we construct assignments $\sigma$ to variables in each $\UGamma$ that satisfy $\UXi$, that is where for each $M = \inl{\tau_1+\tau_2}x_i \in \UXi$ we have $M[\sigma] \to^* \inl{\tau_1+\tau_2}(x_i[\sigma])$. By Theorem \ref{thm:normalisation}, all non-identical $M$ terms are semantically distinct, and as a consequence of \textsc{NFMatch} we will not see identical $M$ terms in the same $\UXi$, so all $\UXi$ are satisfiable.

    We now have a set of assignments $\sigma$ for each branch; the rest of the construction of $X$, where each input is a restriction of some $\sigma$ to just the variables in the starting context, follows by reversing the refinement rules. That $\nref{X}{\UGamma}{\tau}{\top}{N}$ is by induction on the structure of $N$, using the fact that every branch has at least one example in $X$ (by our construction) and that every neutral term can be enumerated (by the lemma). \qed
\end{sketch}

\subsection{Synthesis Algorithm}\label{app:synth-alg}

The algorithm is simply to exhaustively apply the synthesis relations until a solution emerges, or there are no more rules to apply.

Since there are contexts in which $M$-guessing could relate one input to infinitely many terms, meaning our algorithm would get stuck enumerating them, we bound the size of neutral terms that our synthesiser can produce: $n_M$ bounds the neutrals of base type, and $n_\delta$ the neutrals of sum type. With such a limit, synthesis terminates; we then iteratively increase that size limit until a solution is found or a time limit is reached. We also provide a match depth parameter $d$, with the idea of preventing \nsync{} from over-specialising its output to the example set by producing too many branches. Since these parameters bound the search space, we also allow the user to specify a procedure for increasing them, so that synthesis attempts can continue until the parameters are large enough for a solution to be found.

The main function of \nsync{}, shown in Algorithm \ref{alg:nsync}, is just a wrapper around \textbf{NRefine} to handle the parameter search. The mutually recursive functions \textbf{NRefine} and \textbf{PRefine} (Algorithms \ref{alg:nrefine} and \ref{alg:prefine}, respectively) are the key components of \nsync{}, and implement the $N$-refine and $P$-refine relations. \textbf{NRefine} tries to apply \textsc{NRefinePure} first, and falls back on \textbf{NRefineMatch} if that fails; \textbf{PRefine} attempts to apply whichever of the $P$-refinement rules applies to the input based on its goal type. Both functions return \fail{} if no term could be synthesised within the size bounds.

\begin{algorithm}[p]
    \caption{$\textbf{NSynC}(\exset{X}{\UGamma}{\tau}, (n_M,n_\delta,d), \textbf{step})$}\label{alg:nsync}
    \KwData{Example set $\exset{X}{\UGamma}{\tau}$, search parameters $n_M,n_\delta,d \in \mathbb{N}$, search parameter increment procedure $\textbf{step} : \mathbb{N}^3 \to \mathbb{N}^3$}
    \KwResult{A normal term $N$ with no branching arguments which satisfies $X$.}
    $N \gets \textbf{NRefine}(\exset{X}{\UGamma}{\tau};\top;(n_M,n_\delta,d))$\;
    \While{$N = \fail$}{
        $n_M,n_\delta,d \gets \textbf{step}(n_M,n_\delta,d)$\;
        $N \gets \textbf{NRefine}(\exset{X}{\UGamma}{\tau};\top;(n_M,n_\delta,d))$\;
    }
    \Return{$N$}\;
\end{algorithm}

\begin{algorithm}[p]
    \caption{$\textbf{NRefine}(\exset{X}{\UGamma}{\tau};c;(n_M,n_\delta,d))$}\label{alg:nrefine}
    \KwData{Example set $\exset{X}{\UGamma}{\tau}$, constraint $c$, search parameters $n_M,n_\delta,d \in \mathbb{N}$}
    \KwResult{A term $N$ with no branching arguments s.t. $\nref{X}{\UGamma}{\tau}{c}{N}$, or \fail{} if no $N$ exists within the search parameters.}
    $P \gets \textbf{PRefine}(\exset{X}{\UGamma}{\tau}; (n_M,n_\delta,d))$\;
    \lIf{$P \neq \fail$}{\Return{$P$}}
    \ElseIf{$d \geq 1$}{
        \While{$M \gets \textbf{MGuessSum}(\UGamma;c;n_\delta)$}{
            $c \gets (M\sqsubset)$\;
            $\tau_1 + \tau_2 \gets \textbf{type}(M)$\;
            $X_1 \gets \{ \sigma \cdot [v'/x_1] \mapsto v \| \sigma \mapsto v \in X, M[\sigma] \rightarrow^{*} \inl{\tau_1+\tau_2}(v') \}$\;
            $X_2 \gets \{ \sigma \cdot [v'/x_2] \mapsto v \| \sigma \mapsto v \in X, M[\sigma] \rightarrow^{*} \inr{\tau_1+\tau_2}(v') \}$\;
            \lIf{$X_1 = \emptyset \lor X_2 = \emptyset$}{\Skip}
            $N_1 \gets \textbf{NRefine}(\exset{X_1}{\UGamma,x_1:\tau_1}{\tau};(M \sqsubset);(n_M,n_\delta;d-1))$\;
            \lIf{$N_1 = \fail$}{\Skip}
            \lIf{$x_1 \notin FV(N_1) \land \forall \sigma \mapsto v \in X.\: N_1[\sigma] \to^* v$}{\Return{$N_1$}}
            $N_2 \gets \textbf{NRefine}(\exset{X_2}{\UGamma,x_2:\tau_2}{\tau};(M \sqsubset);(n_M,n_\delta;d-1))$\;
            \lIf{$N_2 = \fail$}{\Skip}
            \Return{$\match{M}{x_1.N_1}{x_2.N_2}$}\;
        }
    }
    \Return{\fail}\;
\end{algorithm}

\begin{algorithm}[p]
    \caption{$\textbf{PRrefine}(\exset{X}{\UGamma}{\tau}; (n_M,n_\delta,d))$}\label{alg:prefine}
    \KwData{Example set $\exset{X}{\UGamma}{\tau}$, search parameters $n_M,n_\delta,d \in \mathbb{N}$}
    \KwResult{A term $P$ with no branching arguments s.t. $\pref{X}{\UGamma}{\tau}{P}$, or \fail{} if no $P$ exists within the search parameters.}
    \Match{$\tau$}{
        \uCase{$\theta_i$}{
            $c \gets \top$\;
            \While{$M \gets \textbf{MGuessBase}(\Gamma;\theta_i;c;n_M)$}{
                \lIf{$\forall \sigma \mapsto v \in X.\: M[\sigma] \to^* v$}{\Return{$M$}}
                $c \gets (M \sqsubset)$\;
            }
            \Return{\fail}\;
        }
        \lCase{$1$}{\Return{$\langle\rangle$}\label{lin:ret-unit}}
        \uCase{$\tau_1 \to \tau_2$}{
            $X' \gets \{(\sigma \cdot [v/x]) \mapsto \tilde{f}(v) \| \sigma \mapsto \tilde{f} \in X,\,v \in \text{dom}(\tilde{f})\}$\;
            $N \gets \textbf{NRefine}(\exset{X'}{\UGamma,x:\tau_1}{\tau_2};(x \sqsubseteq);(n_M,n_\delta,d))$\;
            \lIf{$N = \fail$}{\Return{\fail}}
            \Return{$\lam{x:\tau_1}{N}$}\;
        }
        \uCase{$\tau_1 \times \tau_2$}{
            $P_1 \gets \textbf{PRefine}(\exset{\{\sigma \mapsto v_1 \| \sigma\mapsto\pair{v_1}{v_2} \in X\}}{\UGamma}{\tau_1};(n_M,n_\delta,d))$\;
            \lIf{$P_1 = \fail$}{\Return{\fail}}
            $P_2 \gets \textbf{PRefine}(\exset{\{\sigma \mapsto v_2 \| \sigma\mapsto\pair{v_1}{v_2} \in X\}}{\UGamma}{\tau_2};(n_M,n_\delta,d))$\;
            \lIf{$P_2 = \fail$}{\Return{\fail}}
            \Return{$\pair{P_1}{P_2}$}
        }
        \Case{$\tau_1 + \tau_2$}{
            \uIf{$\nexists \sigma \mapsto \inr{\tau}v \in X$}{
                $P \gets \textbf{PRefine}(\exset{\{\sigma \mapsto v \| \sigma \mapsto \inl{\tau}(v) \in X\}}{\UGamma}{\tau_1};(n_M,n_\delta,d)$\;
                \lIf{$P = \fail$}{\Return{\fail}}
                \Return{$\inl{\tau}(P)$}\;
            }
            \uElseIf{$\nexists \sigma \mapsto \inl{\tau}v \in X$}{
                $P \gets \textbf{PRefine}(\exset{\{\sigma \mapsto v \| \sigma \mapsto \inr{\tau}(v) \in X\}}{\UGamma}{\tau_2};(n_M,n_\delta,d)$\;
                \lIf{$P = \fail$}{\Return{\fail}}
                \Return{$\inr{\tau}(P)$}\;
            }
            \Else{
                \Return{\fail}\;
            }
        }
    }
\end{algorithm}

Both of these algorithms call enumeration functions which implement $M$-guessing; \textbf{MGuessSum} and \textbf{MGuessBase}, respectively. Both are wrappers around a successor function $S(\UGamma,M,\tau^\lozenge,n_M)$, which gives the least neutral term $M'$ with no branching arguments s.t. $\UGamma \vdash_\mathrm{M} M':\tau^\lozenge$, $\UGamma \vdash M \sqsubset M'$, and the size of $M'$ is at most $n_M$ (or \fail{} if there is no such term). $\tau^\lozenge$ here is a type template; for example, $\theta_i \times \lozenge$ represents any product type with $\theta_i$ as the first element, and we enumerate neutrals of this type template to get $\UGamma \vdash \fst M:\theta_i$. The full definition of the successor function is omitted for space; it follows from the definition of the ordering relation and the $M$-guessing rules. The definition of \textbf{MGuessSum} is
\begin{align*}
    \textbf{MGuessSum}(\UGamma;(M \sqsubset);n_\delta) &\triangleq S(\UGamma,M,\lozenge+\lozenge,n_\delta) \\
    \textbf{MGuessSum}(\UGamma;(M \sqsubseteq);n_\delta) &\triangleq \begin{cases}
        M & \text{if}\;\UGamma\vdash M:\lozenge+\lozenge \\ & \quad{}\land \text{size}(M) \leq n_\delta \\
        S(\UGamma,M,\lozenge+\lozenge,n_\delta) & \text{otherwise}
    \end{cases} \\
    \textbf{MGuessSum}(\UGamma;\top;n_\delta) &\triangleq \textbf{MGuessSum}(\UGamma;(x_0 \sqsubseteq);n_\delta),
\end{align*}
where in the last case $x_0$ refers to the first variable in $\UGamma$, and \textbf{MGuessSum} returns \fail{} if $\UGamma$ is empty. \textbf{MGuessBase} is definied similarly, but searches for a concrete type $\theta_i$ instead of $\lozenge+\lozenge$.

\subsubsection{Correctness.} As \nsync{} is a faithful implementation of the synthesis relations, we get identical correctness guarantees as long as the parameter search is well-behaved:

\begin{theorem}[Correctness of \nsync{}]\label{thm:nsync-correctness}
    Assume \textbf{step} and $n_M,n_\delta,d$ are such that iterating \textbf{step} on $n_M,n_\delta,d$ tends all three to infinity. Then:
    \begin{enumerate}
        \item \textbf{Soundness.} If $\textbf{NSynC}(\exset{X}{\UGamma}{\tau},(n_M,n_\delta,d),\textbf{step}) = N$, then $N$ satisfies $X$.
        \item \textbf{Semantic Optimality.} If any function of $\textbf{NSynC}$ enumerates a term, it does not enumerate the same term except as part of a different function call.
        \item \textbf{Bounded Completeness.} If there are any $\UGamma \vdash_\mathrm{N} N:\tau$ with no branching arguments that satisfy $X$, then there is one s.t.
        \begin{equation*}
            \textbf{NSynC}(\exset{X}{\UGamma}{\tau}, (n_M, n_\delta, d), \textbf{step}) = N
        \end{equation*}
    \end{enumerate}
\end{theorem}
\begin{sketch}
    We assume (since we omit its definition) correctness of the successor function and take correctness of \textbf{MGuessSum} and \textbf{MGuessBase} as base cases for our induction. In each case, we can easily check that the manipulation of the neutral constraint is correct; $x_0 \sqsubset$ is equivalent to $\top$ because $x_0$ is the least neutral term (a free variable by \textsc{OrdOp}, and the first by \textsc{OrdCtx}), and the rest are trivial. This lets us show that iterating \textbf{MGuessSum} and \textbf{MGuessBase} is equivalent to enumerating the right-hand side of the $M$-guess relation, subject to the size limit.

    We can show by induction that every time \textbf{NRefine} or \textbf{PRefine} returns a term, this term satisfies the corresponding relation; \textbf{MGuessBase} and the unit type case give us base cases for \textbf{PRefine}. Soundness follows from Theorem \ref{thm:relations-soundness}. Semantic optimality requires the extra observation that the only terms considered by any function are those returned by recursive calls, and then follows from Theorem \ref{thm:relations-optimality}.
    
    We can also show by induction that any derivation of $\nref{X}{\UGamma}{\tau}{\top}{N}$ corresponds to a trace of \nsync{}, provided the search parameters are large enough, unless \nsync{} returns some other term $N'$ first, and in that case soundness implies that $\nref{X}{\UGamma}{\tau}{\top}{N'}$ regardless. Bounded completeness follows from Theorem \ref{thm:relations-completeness}. \qed
\end{sketch}

\section{Evaluation}\label{app:evaluation}

We implement \nsync{}---and re-implement \myth{} for comparison---in Haskell with GHC 9.6.7. We conduct our experiments on a 10-core 2.80 GHz Intel i9-10900 with 128GB of RAM, running Ubuntu 22.04.

We generate synthetic benchmarks by randomly generating, in order, a typing context, a goal type inhabitable in that context, a term of that type, and a set of examples for that term. For our preliminary tests, we use four base types, each of which has up to ten distinct values, and produce a typing context of six variables. Each type in the context has a maximum depth of four, and the goal type has a maximum depth of six; we also restrict the goal type to only include the base types that appear in the typing context, otherwise the goal type will be uninhabitable. Our synthesised terms have a maximum neutral size of 10, and a maximum match depth of 5. For each term, we produce 20 examples, and values that are partial functions are generated with six input-output pairs.

Under these parameters, we attempted to generate 200 benchmarks, of which 166 were generated successfully. On the other 32, our generator timed out before finding a term that inhabited the goal type. \myth{}*, our re-implementation, was able to solve 141 of the benchmarks, and \nsync{} solved 104.

We use Haskell's \texttt{criterion} package to time \myth{}* and \nsync{} on each benchmark, with a time limit of 60 seconds for each run. Additionally, we compute the total number of scrutinees enumerated by each algorithm in solving each benchmark. By measuring scrutinees, we can see the number of match terms checked by each algorithm, which we expect to be the majority of the work that \nsync{} eliminates.

Results are presented in the main paper.

\end{document}

\typeout{get arXiv to do 4 passes: Label(s) may have changed. Rerun}